\documentclass[conference]{IEEEtran}
\usepackage{amsmath,graphicx}
\usepackage{color}
\usepackage{graphicx}
\usepackage{epstopdf}
\usepackage{amsmath}
\usepackage{amssymb}
\usepackage{mathrsfs}
\usepackage{tikz}
\usepackage{bm}
\usepackage[english]{babel}
\usepackage{cite}
\usepackage{rotfloat}
\usepackage{mathtools}
\usepackage[font=normalsize,labelfont=bf]{caption}
\usepackage{empheq}

\usepackage{amsmath}
\usepackage{makecell}
\usepackage{algorithm,algorithmic}
\usepackage{multirow}
\usepackage{subfigure}
\usepackage{booktabs}
\usepackage{colortbl}
\usepackage{multirow}
\usepackage{hhline}
\usepackage{stfloats}
\usepackage{multicol}
\usepackage{bbm}
\usepackage{cases}
\graphicspath{ {Figures/} }
\let\oldbibliography\thebibliography
\renewcommand{\thebibliography}[1]{%
  \oldbibliography{#1}%
  \setlength{\itemsep}{1.25pt}%
  \setlength{\baselineskip}{8.25pt}
  \setlength{\lineskiplimit}{-\maxdimen}
}

\newcommand{\T}{{\scriptscriptstyle\mathsf{T}}}
\renewcommand{\H}{{\scriptscriptstyle\mathsf{H}}}

\newsavebox{\foobox}

\definecolor{kugray5}{RGB}{224,224,224}

\usepackage[normalem]{ulem}
\newcommand\rsout{\bgroup\markoverwith
	{\textcolor{red}{\rule[0.5ex]{2pt}{0.8pt}}}\ULon}



\makeatletter
\newcommand{\ALOOP}[1]{\ALC@it\algorithmicloop\ #1%
	\begin{ALC@loop}}
	\newcommand{\ENDALOOP}{\end{ALC@loop}\ALC@it\algorithmicendloop}

\makeatother

\usepackage{etoolbox}
\let\mybibitem\bibitem
\renewcommand{\bibitem}[1]{%
	\ifstrequal{#1}{nature}
	{\color{blue}\mybibitem{#1}}
	{\color{black}\mybibitem{#1}}%
}

\graphicspath{ {Figures/} }

\newtheorem{theorem}{\textbf{Theorem}}

\newtheorem{proof}{Proof}

\newcommand{\epr}{\hfill\(\Box\)}

\DeclareCaptionLabelSeparator{periodspace}{.\quad}

\addto\captionsenglish{}
\newcommand\nbthis{\addtocounter{equation}{1}\tag{\theequation}}

\newcommand{\norm}[1]{\left\lVert#1\right\rVert} 
\newcommand{\fronorm}[1]{\left\lVert#1\right\rVert_{\mathcal{F}}} 
\newcommand{\abs}[1]{\left|#1\right|} 

\newcommand{\tr}[1]{\mathrm{tr}(#1)} 
\newcommand{\trt}[1]{\mathrm{tr}\left(#1\right)} 

\allowdisplaybreaks

\usepackage{setspace}

\newcommand{\mQ}{{\mathbf{Q}}}

\newcommand{\mH}{{\mathbf{H}}} 

\newcommand{\mA}{{\mathbf{A}}}

\newcommand{\mI}{\mathbf{I}}

\newcommand{\mD}{{\mathbf{D}}}
\newcommand{\mX}{{\mathbf{X}}}

\newcommand{\mG}{{\mathbf{G}}}

\newcommand{\mU}{{\mathbf{U}}}
\newcommand{\mV}{{\mathbf{V}}}

\newcommand{\mZ}{{\mathbf{Z}}}

\newcommand{\setC}{\mathbb{C}} 

\newcommand{\vs}{{\mathbf{s}}}

\newcommand{\vu}{{\mathbf{u}}}
 
\newcommand{\vh}{{\mathbf{h}}}

\newcommand{\va}{{\mathbf{a}}}
\newcommand{\vd}{{\mathbf{d}}}

\def\argmin{\mathop{\mathrm{argmin}}}

\def\b0{{\pmb{0}}}

\newcommand{\Pt}{P_{\mathtt{BS}}}
\newcommand{\noise}{\sigma_{\mathrm{n}}^2}

\begin{document}
	\setlength{\abovedisplayskip}{3.0pt}
	\setlength{\belowdisplayskip}{3.0pt}
	\title{Deep Unfolding-Empowered MmWave Massive MIMO Joint Communications and Sensing}

        \author{\IEEEauthorblockN{Nhan~Thanh~Nguyen\IEEEauthorrefmark{1},\! Ly~V.~Nguyen\!\IEEEauthorrefmark{2},\! Nir~Shlezinger\IEEEauthorrefmark{3},\! Yonina~C.~Eldar\IEEEauthorrefmark{4},\! A.~Lee~Swindlehurst\IEEEauthorrefmark{2},\! and Markku~Juntti\IEEEauthorrefmark{1}}
		\IEEEauthorblockA{\IEEEauthorrefmark{1}Centre for Wireless Communications, University of Oulu, P.O.Box 4500, FI-90014, Finland}
		\IEEEauthorblockA{\IEEEauthorrefmark{2}Department of EECS, University of California, Irvine, CA, USA}
		\IEEEauthorblockA{\IEEEauthorrefmark{3}School of ECE, Ben-Gurion University of the Negev, Beer-Sheva, Israel}
            \IEEEauthorblockA{\IEEEauthorrefmark{4}Faculty of Math and CS, Weizmann Institute of Science, Rehovot, Israel}
		
		Emails: \{nhan.nguyen, markku.juntti\}@oulu.fi; \{vanln1, swindle\}@uci.edu, nirshl@bgu.ac.il; yonina.eldar@weizmann.ac.il}
  
  
	
	\maketitle
	
	\begin{abstract}
		In this paper, we propose a low-complexity and fast hybrid beamforming design for joint communications and sensing (JCAS) based on deep unfolding. We first derive closed-form expressions for the gradients of the communications sum rate and sensing beampattern error with respect to the analog and digital precoders. Building on this, we develop a deep neural network as an unfolded version of the projected gradient ascent algorithm, which we refer to as UPGANet. This approach efficiently optimizes the communication-sensing performance tradeoff with fast convergence, enabled by the learned step sizes. UPGANet preserves the interpretability and flexibility of the conventional PGA optimizer while enhancing performance through data training. Our simulations show that UPGANet achieves up to a $33.5\%$ higher communications sum rate and 2.5 dB lower beampattern error compared to conventional designs based on successive convex approximation and Riemannian manifold optimization. Additionally, it reduces runtime and computational complexity by up to $65\%$ compared to PGA without unfolding.
		
	\end{abstract}
	
	\begin{IEEEkeywords}
		Joint communications and sensing, deep unfolding, hybrid beamforming.
	\end{IEEEkeywords}
	\IEEEpeerreviewmaketitle
	
	\section{Introduction}

 	In addition to connectivity, future wireless systems are expected to provide sensing and cognition capabilities~\cite{giordani2020toward}. These emerging systems are often referred to as joint communications and sensing (JCAS)~\cite{zhang2018multibeam}. To generate directional beams and to cope with the severe propagation loss in high-frequency channels, wireless base stations (BSs) will employ large-scale  multiple-input multiple-output (MIMO) arrays, typically implemented via hybrid beamforming (HBF) architectures to meet cost, power, and size constraints~\cite{molisch2017hybrid}.
 	
 	HBF designs for JCAS have gained increasing attention recently. The works in~\cite{qi2022hybrid, wang2022partially, liyanaarachchi2021joint, barneto2021beamformer, cheng2022QoS} optimize radar performance under communications constraints. The approaches in~\cite{qi2022hybrid, cheng2022QoS} minimized the mean squared error (MSE) between the transmit beampattern and a desired pattern, subject to communications signal-to-interference-plus-noise ratio (SINR) and data rate constraints. The studies in~\cite{cheng2021hybrid_narrow} and~\cite{wang2022HBD_OFDM} adopted a different design perspective, optimizing communications performance under radar constraints. To balance radar and communications performance, (weighted) sums of the performance metrics of the two functions are considered in \cite{islam2022integrated, Elbir2021HB_THz, liu2019hybrid, Kaushik2021Hardware,Kaushik2022Green, Nguyen_ISAC_TSP, Nguyen_ML_ISAC_JSTSP}. Specifically, Islam \textit{et al.} aimed to maximize the sum of the communications and radar signal-to-noise ratios (SNRs), while Cheng \textit{et al.}~\cite{cheng2021hybrid} focused on a weighted sum of the communications rate and radar beampattern matching error. In~\cite{Elbir2021HB_THz, liu2019hybrid}, the tradeoff between unconstrained communications and desired radar beamformers is optimized. Kaushik \textit{et al.}~\cite{Kaushik2021Hardware, Kaushik2022Green} addressed RF chain selection to maximize energy efficiency. While these works considered HBF designs for JCAS, they relied on iterative methods that are slow and complex.
 	
 	Data-driven machine learning has emerged as an efficient approach to avoid the high complexity of conventional iterative procedures~\cite{shlezinger2023AI}. In~\cite{Mateos2022EndtoEnd} and~\cite{muth2023Autoencoder}, beamformers, target detection mapping, and receiver processing were jointly performed in a deep end-to-end autoencoder model, focusing on fully digital MIMO with a single receiver. Xu \textit{et al.}~\cite{xu2022deep} utilized deep reinforcement learning to design sparse transmit arrays with quantized phase shifters for HBF with a single RF chain, supporting a single user while operating in either radar or communications mode. Elbir \textit{et al.}~\cite{elbir2021terahertz} trained two convolutional neural networks to estimate the direction of radar targets using a partially connected HBF architecture.
 	
 	However, the previous JCAS HBF designs relied on black-box deep learning architectures, which lack interpretability and require extensive training with large datasets. In contrast, we propose in this paper a model-driven DNN for HBF design in mmWave massive MIMO JCAS systems, based on unfolding the projected gradient ascent (PGA) algorithm. We refer to our approach as UPGANet. Specifically, we first formulate the design objective to balance the tradeoff between communication rate and beampattern quality, and derive closed-form gradients of this objective with respect to (w.r.t.) the analog and digital precoders. Noticing that these gradients exhibit significant imbalances in magnitude, we then develop a modified iterative solver based on the PGA framework to ensure smooth convergence. Finally, we transform this solver into UPGANet by incorporating PGA iterations into layers of a DNN. Numerical studies show that UPGANet efficiently generates hybrid precoders that provide significant improvements in both communications and sensing performance over conventional methods. Moreover, optimizing the step sizes of the PGA algorithm through data training accelerates the convergence of DU-PGA, contributing to its low complexity and run time.

        \section{Signal Model and Problem Formulation}
	\label{sec_system_model}

	
	\subsection{Signal Model}
	We consider a mono-static massive MIMO JCAS system, wherein a BS composed of a uniform linear array simultaneously transmits data to $K$ single-antenna communications users (UEs) and radar signals with a desired beampattern. The BS employs a  fully connected HBF architecture implemented with phase shifters \cite{nguyen2023deep, yu2016alternating}. Let $N$ and $M$ denote the numbers of antennas and RF chains, and let $\mA \in \mathbb{C}^{N \times M}$ and $\mD = [\vd_1, \ldots, \vd_K] \in \mathbb{C}^{M \times K}$ be the analog and digital precoders at the BS, with power constraint $\fronorm{\mA \mD}^2 = \Pt$. Furthermore, denote by $\vs = [s_1, \ldots, s_K] \in {\mathbb{C}}^{K \times 1}$ the data vector transmitted by the BS. Then, the received signal at UE $k$ is given by 
	\begin{align*}
		&y_k = \vh_k^\H \mA \vd_k s_k + \vh_k^\H \sum\nolimits_{k' \neq k}^K \mA \vd_{k'} s_{k'} + n_k, \nbthis \label{processed_received_signal}
	\end{align*}
	where $n_k \sim \mathcal{CN}(0,\noise)$ is  additive white Gaussian noise, and $\vh_k \in \mathbb{C}^{N \times 1}$ is the channel vector from the BS to UE $k$. The mmWave channels $\vh_k$, $k = 1,\ldots,K$, are modeled using the extended Saleh-Valenzuela model \cite{sohrabi2016hybrid, nguyen2019unequally, Nguyen_ISAC_TSP}. We omit details for the channel model due to space constraints.
	
	
	\subsection{Problem Formulation} 
	
	Assume that symbol $s_k$ and digital precoding vector $\vd_k$ are intended for UE $k$. Based on \eqref{processed_received_signal}, the achievable sum rate over all the UEs is given as
	\begin{align*}
		&R = \sum_{k = 1}^K \log_2 \left( 1 + \frac{\abs{\vh_k^\H \mA \vd_k}^2 }{\sum_{k' \neq k}^K \abs{\vh_k^\H \mA \vd_{k'}}^2 + \noise} \right). \nbthis \label{eq_rate}
	\end{align*}
	On the other hand, the quality of the beampattern formed by the hybrid precoders $\{\mA, \mD\}$ can be measured by \cite{liu2018mu}
	\begin{align*}
		\tau &\triangleq  \fronorm{ \mA \mD \mD^\H \mA^\H - \bm{\Psi} }^2 , \nbthis \label{eq_tau}
	\end{align*}
	where $\bm{\Psi} \in \setC^{N \times N}$ is the benchmark signal covariance matrix satisfying \cite{liu2018mu, liu2018toward}:
	\begin{align*}
		\bm{\Psi} = \underset{\substack{ \alpha, \bm{\Psi} \in \mathcal{S} }}{\argmin} \sum\nolimits_{t=1}^{T} \abs{\alpha B_{\mathrm{d}}(\theta_t) - \bar{\va}(\theta_t)^\H  \bm{\Psi}  \bar{\va}(\theta_t)}^2. \nbthis \label{obj_func_beampattern}
	\end{align*}
	Here, $\{\theta_t\}_{t=1}^T$ defines a fine angular grid of $T$ angles that covers the detection range {$[-90^\circ, 90^\circ]$}, {$B_{\mathrm{d}}(\theta_t)$ is the desired beampattern at $\theta_t$}, $\bar{\va}(\theta_t) = [1,e^{j \pi \sin(\theta_t) }, \dots,e^{j (N-1) \pi \sin(\theta_t)}]$ is the steering vector of the transmit array, and $\alpha$ is a scaling factor \cite{liu2018mu}. The feasible space of $\bm{\Psi}$ is $\mathcal{S} = \{\bm{\Psi}:[\bm{\Psi}]_{n,n} = \frac{\Pt}{N}, \forall n, \bm{\Psi} \succeq \bm{0}, \bm{\Psi} = \bm{\Psi}^\H\}$, ensuring that the waveform transmitted by different antennas has the same average transmit power, and $\bm{\Psi}$ is positive semi-definite \cite{liu2018mu}. 
	

	
	
	We aim to design the hybrid precoders $\{\mA, \mD\}$ to optimize the communications--sensing performance tradeoff, which is formulated as:
	\begin{subequations}
		\label{opt_prob}
		\begin{align*} 
			\underset{\substack{ \mA, \mD }}{\textrm{maximize}} \quad & R - \omega \tau \nbthis \label{obj_func} \\
			\textrm{subject to} \quad
			&\abs{[\mA]_{nm}} = 1, \forall n, m, \nbthis \label{cons_analog}\\
			& \fronorm{\mA \mD}^2 = \Pt, \nbthis \label{cons_power}
		\end{align*}
	\end{subequations}
	where constraint \eqref{cons_analog} enforces the unit modulus of the analog precoding coefficients, \eqref{cons_power} is the power constraint, and $[\cdot]_{ij}$ denotes the $(i,j)$-th entry of a matrix. Problem \eqref{opt_prob} is nonconvex and therefore challenging to solve. Specifically, it inherits the constant-modulus constraints of HBF transceiver design \cite{nguyen2019unequally, yu2016alternating, sohrabi2016hybrid} and the strong coupling between the design variables $\mA$ and $\mD$ in \eqref{obj_func} and  \eqref{cons_power}.

	\section{Proposed Design}
	\label{sec_proposed_scheme}
	To address \eqref{opt_prob}, we first present the PGA framework to solve $\mA$ and $\mD$ in an iterative manner. We then propose the UPGANet to accelerate the convergence as well as to improve the performance of the PGA method by leveraging data to cope with the non-convex nature of the problem.  
	
	\subsection{Proposed PGA Optimization Framework}
	\label{sec_PGA}

	We propose leveraging the PGA method in combination with alternating optimization (AO) to solve \eqref{opt_prob}. Specifically, in each iteration, $\mA$ and $\mD$ are found via AO, i.e., one is solved while the other is kept fixed. The solutions to $\mA$ and $\mD$ are then projected onto the feasible space defined by \eqref{cons_analog} and \eqref{cons_power} via normalization. Specifically, for a fixed $\mD$, $\mA$ can be updated at  iteration $i+1$ as: 
	\begin{align*}
		&\mA_{(i+1)} = \mA_{(i)} + \mu_{(i)} \left(\nabla_{\mA} R  - \omega \nabla_{\mA} \tau\right) \Big|_{\mA = \mA_{(i)}}, \nbthis \label{eq_F_PGA} \\
		&[\mA_{(i+1)}]_{nm} = \frac{[\mA_{(i+1)}]_{nm}}{\abs{[\mA_{(i+1)}]_{nm}}},\ \forall n, m, \nbthis \label{eq_projection_F}
	\end{align*}
	where $\nabla_{\mX} f$ is the gradient of function $f$ w.r.t. a complex matrix $\mX$. Similarly, given $\mA$,  $\mD$ can be updated as:
	\begin{align*}
		\mD_{(i+1)} &= \mD_{(i)} + \lambda_{(i)} \left(\nabla_{\mD} R  -  \omega \nabla_{\mD} \tau\right) \Big|_{\mD = \mD_{(i)}}, \nbthis \label{eq_W_PGA}\\
		\mD_{(i+1)} &= \frac{\Pt {\mD}_{(i+1)}}{\fronorm{{\mA}_{(i+1)} \mD_{(i+1)}}}. \nbthis \label{eq_projection_W}
	\end{align*}
	The gradients $\nabla_{\mA} R$ and $\nabla_{\mD} R$ are computed as \cite{nguyen2023fast}
	\begin{align*}
		\nabla_{\mA} R \!&=\! \sum_{k = 1}^K \!\frac{\xi \tilde{\mH}_k \mA \mV}{\tr{ \mA \mV \mA^\H \tilde{\mH}_k} \!+\! \noise} \!-\! \frac{\xi \tilde{\mH}_k \mA \mV_{\bar{k}}}{\tr{ \mA \mV_{\bar{k}} \mA^\H \tilde{\mH}_k} \!+\! \noise}, \nbthis \label{grad_rate_F} \\
		\nabla_{\mD} R \!&=\! \sum_{k = 1}^K\! \frac{\xi \bar{\mH}_k \mD}{\tr{\mD \mD^\H \bar{\mH}_k} \!+\! \noise} \!-\! \frac{\xi \bar{\mH}_k \mD_{\bar{k}}}{  \tr{\mD_{\bar{k}} \mD_{\bar{k}}^\H \bar{\mH}_k} \!+\! \noise} , \nbthis \label{grad_rate_W}
	\end{align*}
	respectively, where $\xi  = 1/\ln 2$, $\mV \triangleq \mD \mD^\H, \mV_{\bar{k}} \triangleq \mD_{\bar{k}} \mD_{\bar{k}}^H, \tilde{\mH}_k \triangleq \vh_k \vh_k^\H, \bar{\mH}_k \triangleq \mA^\H \tilde{\mH}_k \mA$, and $\mD_{\bar{k}} \in \setC^{M \times K}$ is obtained by replacing the $k$-th column of $\mD$ with zeros. 
 
    \begin{theorem}
        \label{theo_gradient_tau}
        The gradients of $\tau$ w.r.t. $\mA$ and $\mD$ are respectively given as
        \begin{align*}
            \nabla_{\mA} \tau &= 2 (\mA \mD \mD^\H \mA^\H - \bm{\Psi}) \mA \mD \mD^\H, \nbthis \label{grad_tau_F} \\
            \nabla_{\mD} \tau &= 2 \mA^\H (\mA \mD \mD^\H \mA^\H - \bm{\Psi}) \mA \mD. \nbthis \label{grad_tau_W}
        \end{align*}
    \end{theorem}
    
    \begin{proof}
        See Appendix \ref{app_proof_grad_tau}. \epr
    \end{proof}

	In the PGA method, $\{\mA, \mD\}$ can be obtained by applying the update rules \eqref{eq_F_PGA} and \eqref{eq_W_PGA}. However, such a straightforward application often yields poor convergence since the gradients of $R$ and $\tau$ w.r.t. $\mA$ and $\mD$ are significantly different in magnitude, which affects their contributions to maximizing $R$ and minimizing $\tau$ at each iteration. To overcome this issue, we propose updating $\mA$ over multiple iterations before updating $\mD$. The approach enables $\mA$ to keep pace with $\mD$ during the PGA iterations. Specifically, $\mA$ is updated as:
	\begin{subnumcases}{}
	 	\hat{\mA}_{(i,0)} = \mA_{(i)}, \label{eq_update_F} \\
	 	\hat{\mA}_{(i,j+1)} = \hat{\mA}_{(i,j)} + \mu_{(i,j)} \left(\nabla_{\mA} R  - \omega \nabla_{\mA} \tau\right) \Big|_{\mA = \hat{\mA}_{(i,j)}},  \label{eq_update_F_inner} \\
	 	\mA_{(i+1)} = \hat{\mA}_{(i,J)},  \label{eq_update_F_outer}
 	\end{subnumcases}
 	where $I$ and $J$ respectively represent the number of outer and inner iterations for updating $\mA$, $i=0,\ldots,I$, and $j=0,\ldots,J-1$. The above update is followed by the projection in \eqref{eq_projection_F}, where $\hat{\mA}_{(i,j)}$ and $\mu_{(i,j)}$ are respectively the precoder and step size in the $j$-th inner iterations of the $i$-th outer iteration, and $\mA_{(i)}$ is the final precoder obtained in the $i$-th outer iteration once all inner iterations have been completed. On the other hand, $\mD$ is updated as
	 \begin{align*}
		 	\mD_{(i+1)} &= \mD_{(i)} + \lambda_{(i)} \left( \nabla_{\mD} R  -  \omega \eta \nabla_{\mD} \tau\right) \Big|_{\mD = \mD_{(i)}}, \nbthis \label{eq_update_W}
		 \end{align*}
	 followed by the projection in \eqref{eq_projection_W}, where $\mD_{(i)}$ is the digital precoder obtained in the $i$-th outer iteration. A weight $\eta$ is applied to $\nabla_{\mD} \tau$ to balance it with $\nabla_{\mA} \tau$. It can be shown that $\abs{\left[\nabla_{\mA} \tau\right]{nm}} \ll \abs{\left[\nabla_{\mD} \tau\right]_{mk}}, \forall n, m$, leading us to set $\eta = \frac{1}{N}$.
	
\subsection{Proposed UPGANet}
\label{ssec:UnfPGA}

\begin{figure}[t]
	\hspace{-0.3cm}
	\includegraphics[scale=0.380]{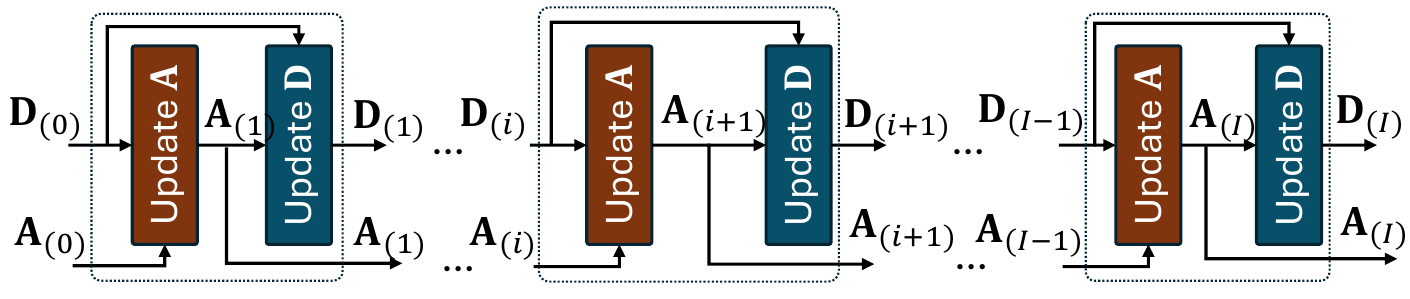}
	\caption{Illustration of UPGANet.}
	\label{fig_unfolding_model}
\end{figure}

To accelerate the convergence of the PGA method without increasing per-iteration complexity due to the backtracking line search, we leverage data to tune the step sizes $\{\mu_{(i,j)},\lambda_{(i)}\}_{i=0,j=0}^{I-1,J-1}$. This is achieved by embedding the proposed PGA optimization framework from Section \ref{sec_PGA} into a DNN, leading to UPGANet. This allows obtaining efficient solutions to $\{\mA, \mD\}$ within a limited and fixed number of iterations, while preserving the interpretability and flexibility of the conventional PGA method. 

We construct UPGANet with $I$ layers by unrolling the $I$ PGA iterations. Its learning task is to generate feasible precoders $\{\mA, \mD\}$ that maximize the objective function $R - \omega\tau$. The unfolding mechanism maps an inner/outer iteration of the PGA procedure to an inner/outer layer of UPGANet. Due to this mapping, we still use subscripts $(i,j)$ for the outer/inner layers when describing UPGANet. For ease of exposition, we denote $\bm{\mu} \triangleq \{\mu_{(i j)}\}_{i, j=0}^{I, J}$ and $\bm{\lambda} \triangleq \{\lambda_{(i)}\}_{i=0}^{I}$. 

UPGANet follows the updating process in \eqref{eq_update_F}--\eqref{eq_update_W}. It takes as input an initial solution $\{\mA_{(0)}, \mD_{(0)}\}$, and data for $\mH = [\vh_1, \ldots, \vh_K]^\H$, $\Pt$, and $\noise$. The output at the $i$-th outer layer is $\{\mA_{(i)}, \mD_{(i)}\}$. Each outer layer includes a sub-network of $J$ layers to output $\mA_{(i)}$, mimicking the process in \eqref{eq_update_F}--\eqref{eq_update_F_outer}. The operations inside each inner/outer layer include computing the gradients in \eqref{grad_rate_W}--\eqref{grad_tau_W} and applying the updating rules \eqref{eq_update_F}--\eqref{eq_update_W} and the projections \eqref{eq_projection_F} and \eqref{eq_projection_W}. UPGANet is illustrated in Fig.\ \ref{fig_unfolding_model}, while its detailed operation will be further discussed in Section \ref{sec_overall_algorithm}.

\subsubsection{Training UPGANet}
Following the design problem \eqref{opt_prob}, UPGANet is trained in an  unsupervised manner to maximize $R - \omega \tau$. Accordingly, the loss function is set to
\begin{align*}
	&\mathcal{L}(\bm{\mu}, \bm{\lambda}) 
	= \omega \fronorm{ \mA_{(I)} \mD_{(I)} \mD_{(I)}^\H \mA_{(I)}^\H - \bm{\Psi} }^2\\
	&\qquad- \sum_{k = 1}^K \log_2 \left( 1 + \frac{\abs{\vh_k^\H \mA_{(I)} \vd_{k(I)}}^2 }{\sum_{k' \neq k}^K \abs{\vh_k^\H \mA_{(I)} \vd_{j(I)}}^2 + \noise} \right). \nbthis \label{eq_loss_0}
\end{align*}
The data set includes multiple channel realizations, and we boost the learned hyperparameters to be suitable for multiple SNRs by randomly choosing $\Pt \in [\gamma_{\min}, \gamma_{\max}]$ dBW while fixing $\noise = 1$. The weight $\omega$ is treated as a given hyperparameter and is chosen to ensure a good tradeoff between $R$ and $\tau$ for moderate-to-high SNRs during training. The loss $\mathcal{L}(\bm{\mu}, \bm{\lambda})$ is a function of the step sizes $\{\bm{\mu}, \bm{\lambda}\}$ because $\{\mA_{(I)}, \mD_{(I)}\}$ depends on $\{\mA_{(i)}\}_{i=0}^{I-1}$, $\{\mD_{(i)}\}_{i=0}^{I-1}$, and $\{\bm{\mu}, \bm{\lambda}\}$. With the loss function \eqref{eq_loss_0}, UPGANet is trained to optimize $\{\bm{\mu}, \bm{\lambda}\}$ to achieve the largest objective value within $I$ iterations. Furthermore, it is important to start the procedure with a good initial solution $\{\mA_{(0)}, \mD_{(0)}\}$. We discuss this issue in the next subsection.

\subsection{Overall Unfolded JCAS Algorithm}
\label{sec_overall_algorithm}
\subsubsection{Overall Algorithm}

\begin{algorithm}[H]
	\small
	\caption{Proposed HBF-JCAS design based on UPGANet}
	\label{alg_algorithm}
	\begin{algorithmic}[1]
		\REQUIRE $\mH$, $\Pt$, $\omega$, and the trained step sizes $\{\bm{\mu}, \bm{\lambda}\}$.
		\ENSURE $\mA$ and $\mD$
		\STATE \textbf{Initialization:} Generate $\{\mA_{(0)}, \mD_{(0)}\}$ based on \eqref{eq_init}.
		
		\FOR{$i = 0 \rightarrow I-1$}
		\STATE Set $\hat{\mA}_{(i,0)} = \mA_{(i)}$.
		\FOR{$j = 0 \rightarrow J - 1$}
		\STATE Obtain the gradients $\nabla_{\mA} R$ and $\nabla_{\mA} \tau$ at $(\mA, \mD) = (\hat{\mA}_{(i,j)}, \mD_{(i)})$ based on \eqref{grad_rate_F} and \eqref{grad_tau_F}.
		\STATE Obtain $\hat{\mA}_{(i,j+1)}$ based on \eqref{eq_update_F_inner}.
		\ENDFOR
		\STATE Set $\mA_{(i+1)} = \hat{\mA}_{(i,J)}$ and apply the projection in \eqref{eq_projection_F}.
		\STATE Obtain the gradients $\nabla_{\mD} R$ and $\nabla_{\mD} \tau$ at $(\mA, \mD) = \left(\mA_{(i+1)}, \mD_{(i)}\right)$ based on \eqref{grad_rate_W} and \eqref{grad_tau_W}.
		\STATE Obtain $\mD_{(i+1)}$ based on \eqref{eq_update_W} and apply the projection \eqref{eq_projection_W}.
		\ENDFOR
		\RETURN $\mA_{(I)}$ and $\mD_{(I)}$ as the solution to $\mA$ and $\mD$.
	\end{algorithmic}
\end{algorithm}

The proposed HBF design based on UPGANet is outlined in Algorithm \ref{alg_algorithm}. {The initial precoders $\{\mA_{(0)}, \mD_{(0)}\}$ are chosen as
\begin{align*}
	[\mA_{(0)}]_{nm} = e^{-j \vartheta_{nm}},\ 
	\mD_{(0)} = \mA_{(0)}^{\dagger} \mX_{\mathrm{ZF}}, \nbthis \label{eq_init}
\end{align*}
with $\mD_{(0)}$ normalized as $\mD_{(0)} = \sqrt{\Pt} \mD_{(0)}/\fronorm{\mA_{(0)} \mD_{(0)}}$ to satisfy \eqref{cons_power}. In \eqref{eq_init}, $\vartheta_{nm}$ is the phase of the $(n,m)$-th entries of $\mG = [\vh_1, \ldots, \vh_K, \bar{\va}(\theta_{\mathtt{d},1}), \ldots, \bar{\va}(\theta_{\mathtt{d},M - K})]$, where $\bar{\va}(\theta_{\mathtt{d},p})$ is the steering vector of the $p$-th desired sensing angle, $p \leq P$. Here, $P$ denotes the number of angles at which the desired beampattern has high gains, and $(\cdot)^{\dagger}$ denotes the pseudo-inverse. Furthermore, we set $\mX_{\mathrm{ZF}} = \mH^\dagger$, assuming that $M \leq K + P$. 
With \eqref{eq_init}, $\mA_{(0)}$ is aligned with the communications channels $\{\vh_1, \ldots, \vh_K\}$ and the sensing steering vectors $\bar{\va}(\theta_{\mathtt{d},p}), \forall p$ to harvest the large array gains. Furthermore, $\mD_{(0)}$ in \eqref{eq_init} is the constrained least-squares solution to the problem $\mathrm{min}_{\mD} \fronorm{\mA_{(0)} \mD - \mX_{\mathrm{ZF}}}$ subject to \eqref{cons_power}. Therefore, the proposed input/initialization can provide good performance in multiuser massive MIMO systems, especially when $N$ is large, as will be further demonstrated in Section \ref{sec_simulation}.}

UPGANet uses the trained step sizes $\{\bm{\mu}, \bm{\lambda}\}$ to perform the updates in \eqref{eq_update_F}--\eqref{eq_update_W} and the projections \eqref{eq_projection_F} and \eqref{eq_projection_W}, as outlined in steps 2--11 of Algorithm \ref{alg_algorithm}. Specifically, steps 3--8 compute the output $\mA_{(i+1)}$ over the $J$ layers. Then, $\mD_{(i+1)}$ is obtained in step 10 based on the updated $\mA_{(i+1)}$. The outcome of the algorithm is the final output of UPGANet.

\subsubsection{Complexity Analysis}

To analyze the complexity of the proposed JCAS-HBF design in Algorithm \ref{alg_algorithm}, we first observe that $\mV$ and $\mV_{\bar{k}}$ are unchanged over $J$ inner iterations, while $\mD$ is of size $(M \times K)$ with $M, K \ll N$. Therefore, the main computational complexity of Algorithm \ref{alg_algorithm} comes from computing the gradients in \eqref{grad_rate_F}, \eqref{grad_tau_F}, \eqref{grad_rate_W}, and \eqref{grad_tau_W}. In \eqref{grad_rate_F}, computing $\mathbf{\tilde{H}}_k\mA$ requires a complexity of $\mathcal{O}(NM)$. Thus, the complexity in computing $\mathbf{\tilde{H}}_k\mA \mV$ is $\mathcal{O}(NM^2)$. Computing $\operatorname{trace}\{\mA \mV \mA^\H \tilde{\mH}_k\}$ requires $\mathcal{O}(NM)$ because $\mV \mA^\H \tilde{\mH}_k = (\mathbf{\tilde{H}}_k\mA \mV)^\H$, where $\mathbf{\tilde{H}}_k\mA \mV$ has been computed as above and the $\tr{\cdot}$ operator only requires the diagonal elements of its matrix argument. Thus, the complexity of the first summation term in \eqref{grad_rate_F} is $\mathcal{O}(NM^2K)$. Since the complexity of the two summation terms in \eqref{grad_rate_F} are the same, the total complexity in computing \eqref{grad_rate_F} is still $\mathcal{O}(NM^2K)$. In \eqref{grad_tau_F}, the matrix $\mathbf{FW}$ is computed first and then  used to compute~\eqref{grad_tau_F}, so the computational complexity required for~\eqref{grad_tau_F} is $\mathcal{O}(N^2K)$. 
Similarly, we can obtain the complexity required for~\eqref{grad_rate_W} and \eqref{grad_tau_W} as $\mathcal{O}(NMK)$ and $\mathcal{O}(N^2K)$, respectively. The overall complexity of Algorithm \ref{alg_algorithm} is thus $\mathcal{O}(IJN^2K)$, as summarized in Table \ref{tab_complexity}.

 \renewcommand{\arraystretch}{1.0}
 \begin{table}[H]
 	\small
 	\begin{center}
 		\caption{\centering Computational complexity of Algorithm \ref{alg_algorithm}.}
 		\label{tab_complexity}
 		\begin{tabular}{|c|c|}
 			\hline
 			Tasks & Complexities \\
 			\hline
 			\hline
 			Compute $\nabla_{\mA} R$ & $\mathcal{O}(NM^2K)$ (per inner iteration/layer) \\
 			\hline
 			Compute $\nabla_{\mD} R$ & $\mathcal{O}(N^2K)$ (per inner iteration/layer) \\
 			\hline 
 			Compute $\nabla_{\mA} \tau$ & $\mathcal{O}(NMK)$  (per outer iteration/layer) \\
 			\hline 
 			Compute $\nabla_{\mA} \tau$ & $\mathcal{O}(N^2K)$ (per outer iteration/layer) \\
 			\hline
 			\textbf{Overall algorithm} & $\mathcal{O}(IJN^2K)$ \\
 			\hline
 		\end{tabular}
 	\end{center}
 \end{table}

\section{Simulation Results}
\label{sec_simulation}

Here we provide numerical results to demonstrate the performance of the proposed JCAS-HBF designs. Unless otherwise stated, we assume scenarios with $P=3$, $K = M = 4$, $N=64$, and $\omega = 0.3$. The mmWave communications channels are generated using the same parameters as in \cite{sohrabi2016hybrid}. We assume a desired beampattern that steers beams towards the following $P=3$ directions: $\{-60^{\circ}, 0^{\circ}, 60^{\circ}\}$. Thus, the desired beampattern is defined as $B_{\mathtt{d}}(\theta_t) = 1$ for $\theta_t \in [\theta_{\mathtt{d},p} - \delta_{\theta}, \theta_{\mathtt{d},p} + \delta_{\theta}]$; otherwise, $B_{\mathtt{d}}(\theta_t) = 0$  \cite{cheng2021hybrid}. Here, ${\delta_{\theta} = 5^\circ}$ is half the mainlobe beamwidth of $B_{\mathtt{d}}(\theta_t)$. UPGANet is implemented using Python with the Pytorch library. It is trained for $I=120$ and the SNR range $[\gamma_{\min}, \gamma_{\max}] = [0, 12]$ dB using the Adam optimizer with $1000$ channels over $100$ and $30$ epochs for $J = 1$ and $J = \{10,20\}$, respectively. We set $\mu_{(0,0)} = \lambda_{(0)} = 0.01$, which are also used as the fixed step sizes for the PGA algorithm without unfolding.

\begin{figure}
	\centering
	\includegraphics[scale=0.5]{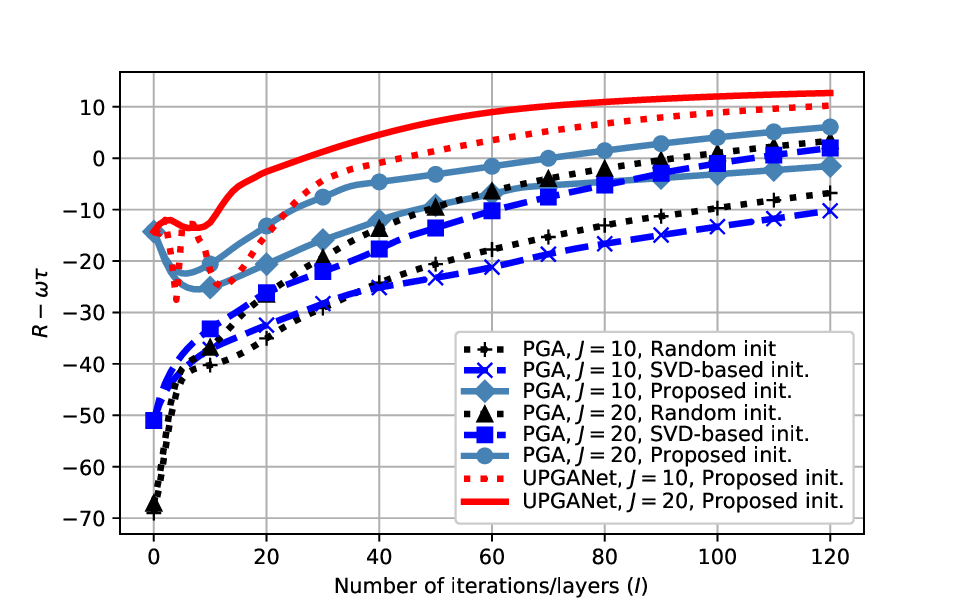}
	\caption{Convergence of the PGA algorithm with $N=64, K = M = 4, J = \{10, 20\}, \omega = 0.3$, SNR $= 12$ dB, and different initializations.}
	\label{fig_comp_conv_init}
\end{figure}

\begin{figure*}[!htb]
	\hspace{-0.5cm}
	\subfigure[Convergence (SNR $= 12$ dB)]
	{
		\includegraphics[scale=0.41]{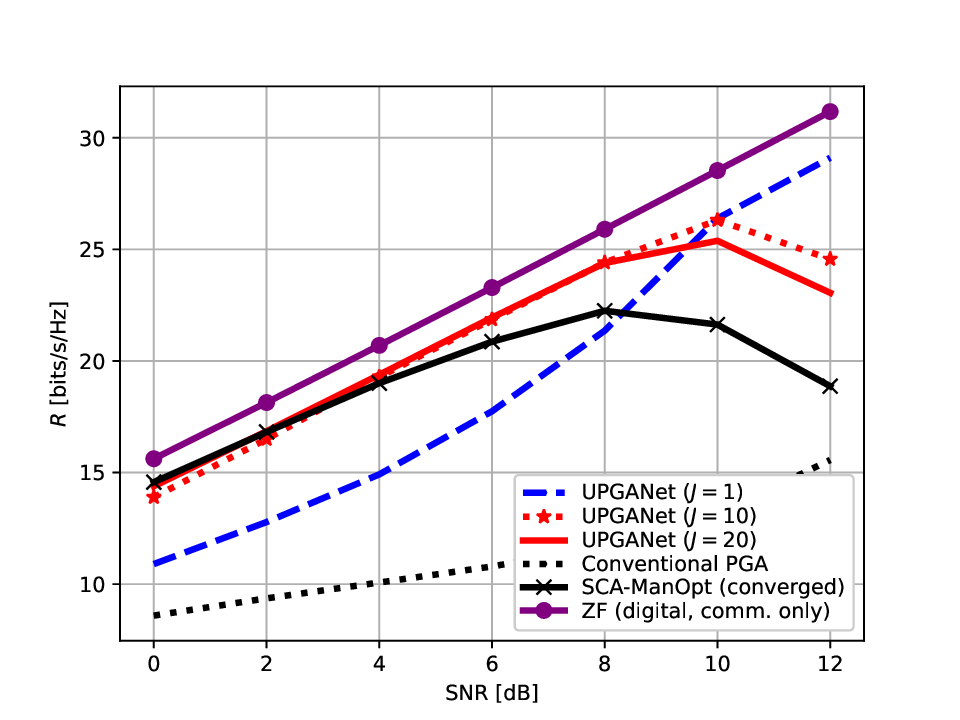}
		\label{fig_rate_vs_SNR_64}
	}\hspace{-0.95cm}
	\subfigure[Sum rate]
	{
		\includegraphics[scale=0.41]{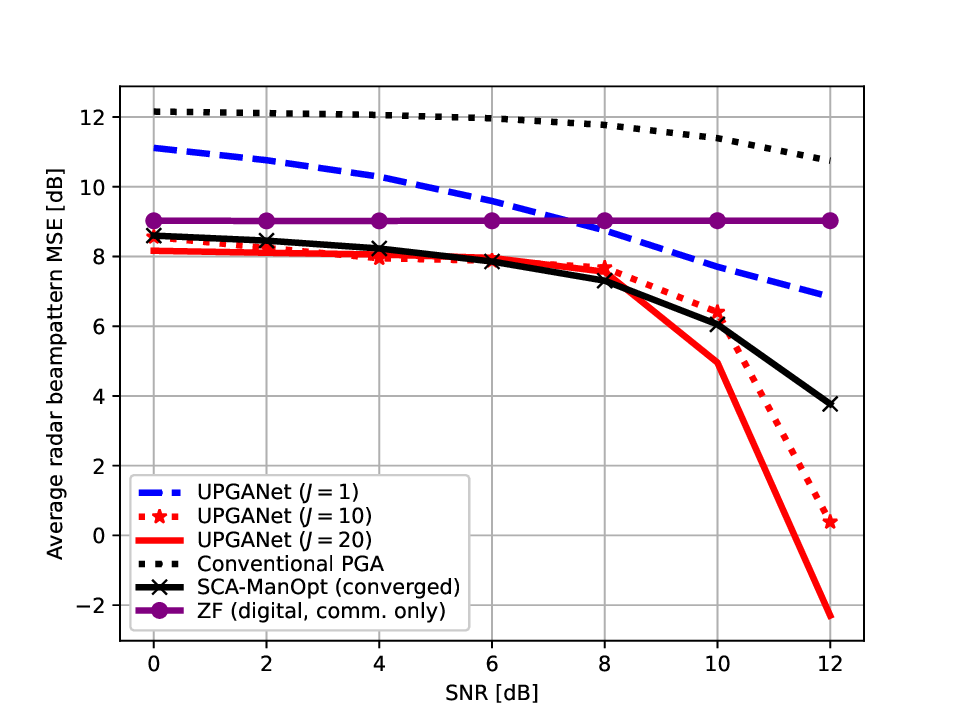}
		\label{fig_MSE_vs_SNR_64}
	}
	\hspace{-0.95cm}
	\subfigure[Average beampattern MSE]
	{
		\includegraphics[scale=0.41]{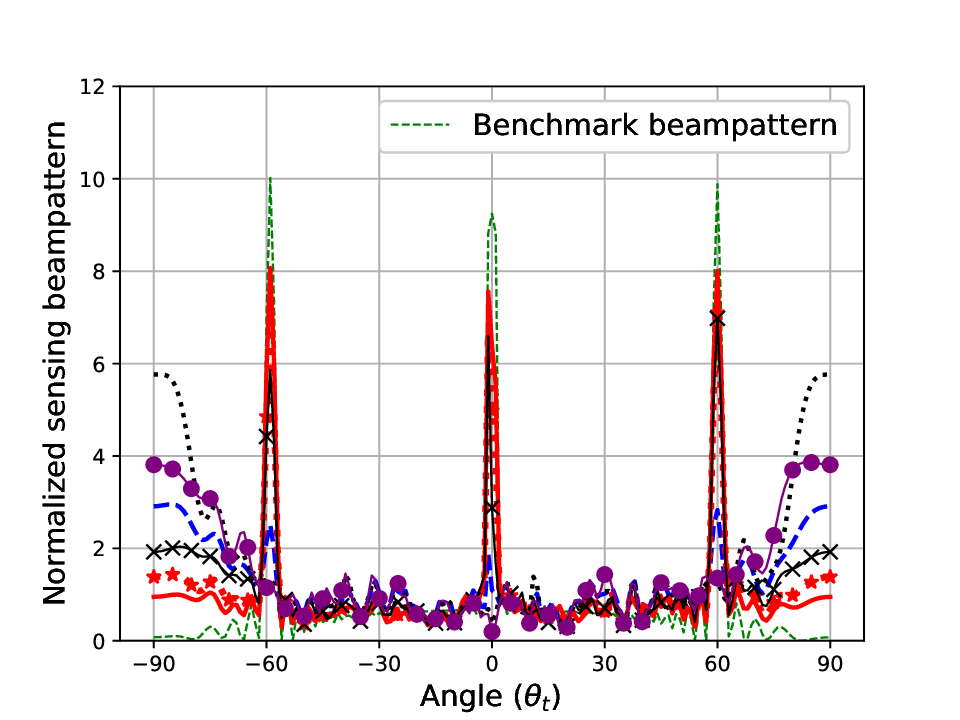}
		\label{fig_beampattern_64}
	}
	\caption{Performance and beampattern of the considered schemes versus SNRs with $N=64$, $K = M = 4$, $\omega = 0.3$, and $J = \{1, 10, 20\}$.}\vspace{-0.55cm}
	\label{fig_SNR_64}
\end{figure*}

In Fig.\ \ref{fig_comp_conv_init}, we demonstrate the convergence of UPGANet with $J = \{10,20\}$ as well as the effect of the initial solution/input in \eqref{eq_init} on the convergence. For comparison, we consider the following methods: \textit{Random init:} We randomly generate $\mA_{(0)}$ and set  $\mD_{(0)} = \left( \mH \mA_{(0)}\right)^{\dagger}$ \cite{sohrabi2016hybrid}; \textit{SVD-based init:} We employ $\mA_{(0)} = [\vu_1, \ldots, \vu_K, \va(\theta_{\mathtt{d},1}), \ldots, \va(\theta_{\mathtt{d},M - K})]$, where $\vu_k$ is the $k$-th principle singular vector of $\mH$, and $\bar{\va}(\theta_{\mathtt{d},p})$ is the steering vector \cite{agiv2022learn}, while $\mD_{(0)}$ is the same as in \eqref{eq_init}. It is observed from Fig.\ \ref{fig_comp_conv_init} that the proposed initialization substantially improves the convergence of PGA with and without unfolding. Specifically, $\{\mA_{(0)}, \mD_{(0)}\}$ in \eqref{eq_init} yields both a higher initial and final objective values than the other initializations. The SVD-based method yields a relatively good initial objective, but after a few iterations, it behaves similarly to the random initialization, which has not converged after $120$ iterations. It is also seen that a larger $J$ leads to better convergence, and UPGANet with $J=20$ exhibits the best convergence. 

The complexity and runtime reduction of UPGANet are also observed in Fig.\ \ref{fig_comp_conv_init}. Comparing (i) UPGANet with $J=20$, (ii) PGA with $J=20$, and (iii) PGA with $J=10$, all of which employ the proposed initialization, we see that these schemes reach $R - \omega \tau = 0$ at $I \approx \{25, 70, 120\}$, respectively. This means that to achieve the peak value of the objective, approach (i) requires $n_{\mA} = IJ = 25 \times 20 = 500$ updates of $\mA$ and $n_{\mD} = I = 25$ updates of $\mD$. On the other hand, algorithms (ii) and (iii) require $(n_{\mA}, n_{\mD}) = (1400, 70)$ and $(n_{\mA}, n_{\mD}) = (1200, 120)$ updates, respectively.

To demonstrate the communications and sensing performance of UPGANet, we consider the following benchmark approaches: (i) \textit{conventional PGA}: We fix the step sizes as $\mu_{(0)} = \lambda_{(0)} = 0.01$ and set $J=1$. (ii) \textit{SCA-ManOpt}: We first find the precoder $\mX^{\star}$ that maximizes $R$ via the SCA method \cite{tran2012fast}. Then, $\mX$ is obtained by maximizing $\rho \fronorm{ \mX - \mX^{\star} }^2 + (1-\rho) \fronorm{ \mX \mX^\H - \bm{\Psi} }^2$ with $\rho = 0.2$ \cite{liu2019hybrid, liu2018mu, liu2018toward}. After that, we obtain $\{\mA, \mD\}$ via matrix factorization \cite{yu2016alternating}. The convergence tolerances of both the SCA and ManOpt procedures are set to $\varepsilon = 10^{-3}$. (iii) \textit{ZF (digital, comm. only)}: We employ a fully digital ZF beamformer in the communications-only system, which has near-optimal communications performance \cite{zhang2018performance} and serves as an upper bound on the sum rate achieved by the JCAS-HBF approaches.

In Figs.\ \ref{fig_rate_vs_SNR_64} and \ref{fig_MSE_vs_SNR_64}, we show the communications sum rate $R$ and the average radar beampattern mean square error (MSE) of the considered approaches. The beampattern MSE is computed as $\text{MSE} = \frac{1}{T} \sum_{t=1}^{T} \abs{B_{\mathtt{d}}(\theta_t) - \bar{\va}^\H(\theta_t)  \bm{\Psi}  \bar{\va}(\theta_t)}^2$. It is observed from the figure that UPGANet with $J=\{5, 20\}$ performs close to the communications-only system with the ZF beamformer and surpasses SCA-ManOpt in terms of communications sum rate, while maintaining comparable or lower radar beampattern MSEs, especially at high SNR. For instance, at SNR $= 12$ dB, UPGANet with $J=\{10, 20\}$ achieves improvements of about $\{33.2\%, 24.7\%\}$ in sum rate and $\{2.5, 6\}$ dB lower MSEs compared with SCA-ManOpt, respectively. UPGANet with $J=1$ can attain good communications performance at high SNR; however, it has degraded sensing performance. Among the compared schemes, the conventional PGA with $J=1$ has the worst performance. 

We note the performance loss of all the considered JCAS approach in the interference-limited regime, as clearly seen at high SNRs. In this regime, the JCAS waveform needs to serve both the radar targets and the communications users at the same time. Thus, its capability to mitigate inter-user interference (IUI) is significantly reduced, causing decreased $R$ in interference-limited cases or when using a larger $\omega$. We note that this can be avoided by setting a smaller $\omega$ or normalizing $\tau$; however, this adjustment leads to compromised sensing performance.

Finally, we compare the beampattern of UPGANet to those of the considered benchmarks in Fig.\  \ref{fig_beampattern_64}. It is observed that the average sensing beampatterns obtained by UPGANet fit the benchmark beampattern the best. They have significantly higher peaks at the target angles $\{-60^{\circ}, 0^{\circ}, 60^{\circ}\}$ and lower side lobes compared to the beampattern obtained with SCA-ManOpt and conventional PGA. As expected, the beampattern generated by the communication-only design with the ZF precoder is the worst because it is designed to maximize the communications rate.

\section{Conclusions}
\label{sec_conclusion}

We have investigated multiuser massive MIMO JCAS systems with HBF transceiver architectures, aiming at maximizing the tradeoff between the communications sum rate and the radar sensing beampattern accuracy. By deriving and analyzing the gradients of those metrics, we proposed effective updating rules for the analog and digital precoders for improved convergence of the PGA optimization. We further proposed UPGANet based on the deep unfolding technique, where the step sizes of the PGA approach are learned in an unsupervised manner. UPGANet performs much faster with significantly reduced computational complexity thanks to its well-trained step sizes compared to the unfolded version. Our extensive numerical results demonstrate that UPGANet achieves significant improvements in communications and sensing performance w.r.t. conventional JCAS-HBF designs.

\appendices

\section{Proof of Theorem \ref{theo_gradient_tau}}

\label{app_proof_grad_tau}
We first recall that $\left[\nabla_{\mZ} f\right]_{rc} = \frac{\partial f}{\partial [\mZ]_{rc}^*}$, where $\mZ \in \setC^{R \times C}$. Thus, $\nabla_{\mA} \tau$ and $\nabla_{\mD} \tau$ can be obtained using $\partial \tau / \partial [\mA]_{nm}^*$ and $\partial \tau / \partial [\mD]_{mk}^*$, respectively, with $n = 1,\ldots,N$, $m = 1,\ldots,M$, and $k = 1,\ldots, K$. Let us denote $\mU \triangleq \mA \mD \mD^\H \mA^\H$ and rewrite $\tau$ as $\tau = \norm{ \mU - \bm{\Psi} }_{\mathcal{F}}^2$. Then, $\partial \tau / \partial [\mA]_{nm}^*$ and $\partial \tau / \partial [\mD]_{mk}^*$ can be derived by the following chain rule
\begin{align*}
	\frac{\partial \tau}{\partial [\mA]_{nm}^*} &= \trt{\left( \frac{\partial \tau}{\partial \mU^\H} \right)^\T \frac{\partial \mU^\H}{\partial [\mA]_{nm}^*} } = \trt{\frac{\partial \tau}{\partial \mU^*} \frac{\partial \mU}{\partial [\mA]_{nm}^*} }, \nbthis \label{eq_grad_tau_F_2}\\
	\frac{\partial \tau}{\partial [\mD]_{mk}^*} &= \trt{\left( \frac{\partial \tau}{\partial \mU^\H} \right)^\T \frac{\partial \mU^\H}{\partial [\mD]_{mk}^*} } = \trt{\frac{\partial \tau}{\partial \mU^*} \frac{\partial \mU}{\partial [\mD]_{mk}^*} }, \nbthis \label{eq_grad_tau_W_2}
\end{align*}
where \eqref{eq_grad_tau_F_2} and \eqref{eq_grad_tau_W_2} follow from the fact that $\mU = \mU^\H$.

To compute $\partial \tau/\partial \mU^*$, we rewrite $\tau = \tr{\mU \mU^\H - \bm{\Psi} \mU^\H - \mU \bm{\Psi}^\H + \bm{\Psi} \bm{\Psi}^\H}$
and note that since $\partial \tr{\mU \mU^\H}/\partial \mU^* = 2 \mU$ and $\partial \tr{\mU \bm{\Psi}^\H}/\partial \mU^* = \partial \tr{\bm{\Psi} \mU^\H}/\partial \mU^* = \bm{\Psi}$ \cite{petersen2008matrix}, we have
\begin{align*}
	\partial \tau / \partial \mU^* = 2(\mU - \bm{\Psi}). \nbthis \label{eq_grad_tau_R}
\end{align*}

We now compute $\partial \mU/\partial [\mA]_{nm}^*$ in \eqref{eq_grad_tau_F_2}. Let us write $[\mU]_{ij} = \tr{\bm{\delta}_i^\H \mA \mD \mD^\H \mA^\H \bm{\delta}_j } = \tr{\mA \mD \mD^\H \mA^\H \bm{\delta}_j \bm{\delta}_i^\H}$ where $\bm{\delta}_i$ and $\bm{\delta}_j$ are the $i$-th and $j$-th columns of identity matrix $\mI_N$, respectively. Then, using the result $ \partial \tr{\mZ \mQ_0 \mZ^\H \mQ_1} / \partial \mZ^* = \mQ_1 \mZ \mQ_0$ in \cite{hjorungnes2007complex}, we have $\partial [\mU]_{ij} / \partial \mA^* = \bm{\delta}_j \bm{\delta}_i^\H \mA \mD \mD^\H$. 
Furthermore, since $\partial [\mU]_{ij} / \partial [\mA]_{nm}^*$ is the $(n, m)$-th entry of $\partial [\mU]_{ij} / \partial \mA^*$, we can write
\begin{align*}
	\partial [\mU]_{ij}/{\partial [\mA]_{nm}^*} = \bm{\delta}_n^\H \bm{\delta}_j \bm{\delta}_i^\H \mA \mD \mD^\H \bm{\delta}_m = \bm{\delta}_i^\H \mA \mD \mD^\H \bm{\delta}_m \bm{\delta}_n^\H \bm{\delta}_j ,  \nbthis \label{eq_grad_tau_F_4}
\end{align*}
where $\bm{\delta}_n$ and $\bm{\delta}_m$ are the $n$-th and $m$-th columns of identity matrices $\mI_N$ and $\mI_{M}$, respectively. The second equality in \eqref{eq_grad_tau_F_4} holds because $\bm{\delta}_n^\H \bm{\delta}_j$ is a scalar. Thus, we have
\begin{align*}
	\partial \mU/{\partial [\mA]_{nm}^*} = \mA \mD \mD^\H \bm{\delta}_m \bm{\delta}_n^\H.  \nbthis \label{eq_grad_tau_F_5}
\end{align*}
By substituting \eqref{eq_grad_tau_R} and \eqref{eq_grad_tau_F_5} into \eqref{eq_grad_tau_F_2}, we obtain
\begin{align*}
	\partial \tau/{\partial [\mA]_{nm}^*} &= 2 \tr{(\mU-\bm{\Psi}) \mA \mD \mD^\H \bm{\delta}_m \bm{\delta}_n^\H}\\
	&= 2 \bm{\delta}_n^\H(\mU-\bm{\Psi}) \mA \mD \mD^\H \bm{\delta}_m.  \nbthis \label{eq_grad_tau_F_6}
\end{align*}
Again, we utilize the fact that $\partial \tau/\partial [\mA]_{nm}^*$ is the $(n,m)$-th element of $\partial \tau/\partial [\mA]_{nm}^*$ to obtain
\begin{align*}
	\partial \tau/{\partial \mA^*} &= 2 (\mU-\bm{\Psi}) \mA \mD \mD^\H.  \nbthis \label{eq_grad_tau_F_7}
\end{align*}
Replacing $\mU$ by $\mA \mD \mD^\H \mA^\H$ in \eqref{eq_grad_tau_F_7} gives us the result \eqref{grad_tau_F}.

The derivation of $\partial \tau/\partial \mD^*$ follows a similar approach. For brevity, we omit the detailed derivations here.

\bibliographystyle{IEEEtran}
\bibliography{IEEEabrv,Bibliography}

\begin{thebibliography}{10}
\providecommand{\url}[1]{#1}
\csname url@samestyle\endcsname
\providecommand{\newblock}{\relax}
\providecommand{\bibinfo}[2]{#2}
\providecommand{\BIBentrySTDinterwordspacing}{\spaceskip=0pt\relax}
\providecommand{\BIBentryALTinterwordstretchfactor}{4}
\providecommand{\BIBentryALTinterwordspacing}{\spaceskip=\fontdimen2\font plus
\BIBentryALTinterwordstretchfactor\fontdimen3\font minus
  \fontdimen4\font\relax}
\providecommand{\BIBforeignlanguage}[2]{{%
\expandafter\ifx\csname l@#1\endcsname\relax
\typeout{** WARNING: IEEEtran.bst: No hyphenation pattern has been}%
\typeout{** loaded for the language `#1'. Using the pattern for}%
\typeout{** the default language instead.}%
\else
\language=\csname l@#1\endcsname
\fi
#2}}
\providecommand{\BIBdecl}{\relax}
\BIBdecl

\bibitem{giordani2020toward}
M.~Giordani, M.~Polese, M.~Mezzavilla, S.~Rangan, and M.~Zorzi, ``{Toward 6G
  networks: Use cases and technologies},'' \emph{{IEEE} Commun. Mag.}, vol.~58,
  no.~3, pp. 55--61, 2020.

\bibitem{zhang2018multibeam}
J.~A. Zhang, X.~Huang, Y.~J. Guo, J.~Yuan, and R.~W. Heath, ``Multibeam for
  joint communication and radar sensing using steerable analog antenna
  arrays,'' \emph{{IEEE} Trans. Veh. Technol.}, vol.~68, no.~1, pp. 671--685,
  2018.

\bibitem{molisch2017hybrid}
A.~F. Molisch, V.~V. Ratnam, S.~Han, Z.~Li, S.~L.~H. Nguyen, L.~Li, and
  K.~Haneda, ``Hybrid beamforming for massive {MIMO}: A survey,'' \emph{{IEEE}
  Commun. Mag.}, vol.~55, no.~9, pp. 134--141, 2017.

\bibitem{qi2022hybrid}
C.~Qi, W.~Ci, J.~Zhang, and X.~You, ``Hybrid beamforming for millimeter wave
  {MIMO} integrated sensing and communications,'' \emph{{IEEE} Commun. Lett.},
  vol.~26, no.~5, pp. 1136--1140, 2022.

\bibitem{wang2022partially}
X.~Wang, Z.~Fei, J.~A. Zhang, and J.~Xu, ``Partially-connected hybrid
  beamforming design for integrated sensing and communication systems,''
  \emph{{IEEE} Trans. Commun.}, vol.~70, no.~10, pp. 6648--6660, 2022.

\bibitem{liyanaarachchi2021joint}
S.~D. Liyanaarachchi, C.~B. Barneto, T.~Riihonen, M.~Heino, and M.~Valkama,
  ``Joint multi-user communication and {MIMO} radar through full-duplex hybrid
  beamforming,'' in \emph{IEEE Int. Online Symposium Joint Commun. \& Sensing
  (JC\&S)}, 2021.

\bibitem{barneto2021beamformer}
C.~B. Barneto \emph{et~al.}, ``Beamformer design and optimization for joint
  communication and full-duplex sensing at mm-{W}aves,'' \emph{IEEE Trans.
  Commun.}, vol.~70, no.~12, pp. 8298--8312, Dec. 2022.

\bibitem{cheng2022QoS}
Z.~Cheng and B.~Liao, ``{Q}o{S}-aware hybrid beamforming and {DOA} estimation
  in multi-carrier dual-function radar-communication systems,'' \emph{IEEE J.
  Select. Areas in Commun.}, vol.~40, no.~6, pp. 1890--1905, Sept. 2022.

\bibitem{cheng2021hybrid_narrow}
Z.~Cheng, Z.~He, and B.~Liao, ``Hybrid beamforming for multi-carrier
  dual-function radar-communication system,'' \emph{{IEEE} Trans. on Cogn.
  Commun. Netw.}, vol.~7, no.~3, pp. 1002--1015, 2021.

\bibitem{wang2022HBD_OFDM}
B.~Wang, Z.~Cheng, L.~Wu, and Z.~He, ``Hybrid beamforming design for {OFDM}
  dual-function radar-communication system with double-phase-shifter
  structure,'' in \emph{Proc. European Signal Process. Conf.}, Belgrade,
  Serbia, Aug. 2022, pp. 1067--1071.

\bibitem{islam2022integrated}
M.~A. Islam, G.~C. Alexandropoulos, and B.~Smida, ``Integrated sensing and
  communication with millimeter wave full duplex hybrid beamforming,'' in
  \emph{Proc. IEEE Int. Conf. Commun.}, 2022, pp. 4673--4678.

\bibitem{Elbir2021HB_THz}
A.~M. Elbir, K.~V. Mishra, and S.~Chatzinotas, ``Hybrid beamforming for
  {T}erahertz joint ultra-massive {MIMO} radar-communications,'' in \emph{Proc.
  Int. Symp. on Wireless Commun. Systems}, Berlin, Germany, Sept. 2021.

\bibitem{liu2019hybrid}
F.~Liu and C.~Masouros, ``{Hybrid beamforming with sub-arrayed MIMO radar:
  Enabling joint sensing and communication at mmWave band},'' in \emph{Proc.
  IEEE Int. Conf. on Acoustics, Speech and Signal Process.}, Brighton, UK, May
  2019, pp. 7770--7774.

\bibitem{Kaushik2021Hardware}
A.~Kaushik, C.~Masouros, and F.~Liu, ``Hardware efficient joint
  radar-communications with hybrid precoding and {RF} chain optimization,'' in
  \emph{Proc. IEEE Int. Conf. Commun.}, Montreal, QC, Canada, June 2021.

\bibitem{Kaushik2022Green}
A.~Kaushik, E.~Vlachos, C.~Masouros, C.~Tsinos, and J.~Thompson, ``Green joint
  radar-communications: {RF} selection with low resolution {DAC}s and hybrid
  precoding,'' in \emph{Proc. IEEE Int. Conf. Commun.}, Seoul, Republic of
  Korea, May 2022, pp. 3160--3165.

\bibitem{Nguyen_ISAC_TSP}
N.~T. Nguyen, N.~Shlezinger, Y.~C. Eldar, and M.~Juntti, ``Multiuser {MIMO}
  wideband joint communication and sensing system with subcarrier allocation,''
  \emph{{IEEE} Trans. Signal Process.}, vol.~71, pp. 2997--3013, 2023.

\bibitem{Nguyen_ML_ISAC_JSTSP}
N.~T. Nguyen \emph{et~al.}, ``Joint communication and sensing hybrid
  beamforming design via deep unfolding,'' \emph{{IEEE} J. Sel. Topics Signal
  Process.}, 2024, early access.

\bibitem{cheng2021hybrid}
Z.~Cheng, Z.~He, and B.~Liao, ``Hybrid beamforming design for {OFDM}
  dual-function radar-communication system,'' \emph{{IEEE} J. Sel. Topics
  Signal Process.}, vol.~15, no.~6, pp. 1455--1467, 2021.

\bibitem{shlezinger2023AI}
N.~Shlezinger, M.~Ma, O.~Lavi, N.~T. Nguyen, Y.~C. Eldar, and M.~Juntti,
  ``{AI}-empowered hybrid {MIMO} beamforming,'' \emph{{IEEE} Veh. Technol.
  Mag.}, 2024, early access.

\bibitem{Mateos2022EndtoEnd}
J.~M. Mateos-Ramos, J.~Song, Y.~Wu, C.~Häger, M.~F. Keskin, V.~Yajnanarayana,
  and H.~Wymeersch, ``End-to-end learning for integrated sensing and
  communication,'' in \emph{Proc. IEEE Int. Conf. Commun.}, Seoul, Republic of
  Korea, May 2022, pp. 1942--1947.

\bibitem{muth2023Autoencoder}
C.~Muth and L.~Schmalen, ``Autoencoder-based joint communication and sensing of
  multiple targets,'' in \emph{Proc. Int. ITG Workshop on Smart Antennas and
  Conf. on Systems, Commun., and Coding}, Feb. 2023.

\bibitem{xu2022deep}
L.~Xu, R.~Zheng, and S.~Sun, ``A deep reinforcement learning approach for
  integrated automotive radar sensing and communication,'' in \emph{Proc. IEEE
  Sensor Array and Multichannel Sign. Proc. Workshop}, 2022, pp. 316--320.

\bibitem{elbir2021terahertz}
A.~M. Elbir, K.~V. Mishra, and S.~Chatzinotas, ``Terahertz-band joint
  ultra-massive {MIMO} radar-communications: {M}odel-based and model-free
  hybrid beamforming,'' \emph{{IEEE} J. Sel. Topics Signal Process.}, vol.~15,
  no.~6, pp. 1468--1483, 2021.

\bibitem{nguyen2023deep}
N.~T. Nguyen, M.~Ma, O.~Lavi, N.~Shlezinger, Y.~C. Eldar, A.~L. Swindlehurst,
  and M.~Juntti, ``Deep unfolding hybrid beamforming designs for {THz} massive
  {MIMO} systems,'' \emph{{IEEE} Trans. Signal Process.}, vol.~71, pp.
  3788--3804, 2023.

\bibitem{yu2016alternating}
X.~Yu, J.-C. Shen, J.~Zhang, and K.~B. Letaief, ``Alternating minimization
  algorithms for hybrid precoding in millimeter wave {MIMO} systems,''
  \emph{{IEEE} J. Sel. Topics Signal Process.}, vol.~10, no.~3, pp. 485--500,
  2016.

\bibitem{sohrabi2016hybrid}
F.~Sohrabi and W.~Yu, ``Hybrid digital and analog beamforming design for
  large-scale antenna arrays,'' \emph{{IEEE} J. Sel. Topics Signal Process.},
  vol.~10, no.~3, 2016.

\bibitem{nguyen2019unequally}
N.~T. Nguyen and K.~Lee, ``Unequally sub-connected architecture for hybrid
  beamforming in massive {MIMO} systems,'' \emph{{IEEE} Trans. Wireless
  Commun.}, vol.~19, no.~2, pp. 1127--1140, 2019.

\bibitem{liu2018mu}
F.~Liu, C.~Masouros, A.~Li, H.~Sun, and L.~Hanzo, ``{MU-MIMO communications
  with MIMO radar: From co-existence to joint transmission},'' \emph{{IEEE}
  Trans. Wireless Commun.}, vol.~17, no.~4, pp. 2755--2770, 2018.

\bibitem{liu2018toward}
F.~Liu, L.~Zhou, C.~Masouros, A.~Li, W.~Luo, and A.~Petropulu, ``Toward
  dual-functional radar-communication systems: Optimal waveform design,''
  \emph{{IEEE} Trans. Signal Process.}, vol.~66, no.~16, pp. 4264--4279, 2018.

\bibitem{nguyen2023fast}
N.~T. Nguyen \emph{et~al.}, ``Fast deep unfolded hybrid beamforming in
  multiuser large {MIMO} systems,'' in \emph{Proc. Annual Asilomar Conf.
  Signals, Syst., Comp.}, 2023, pp. 486--490.

\bibitem{agiv2022learn}
O.~Agiv and N.~Shlezinger, ``Learn to rapidly optimize hybrid precoding,'' in
  \emph{Proc. IEEE Works. on Sign. Proc. Adv. in Wirel. Comms.}, 2022.

\bibitem{tran2012fast}
L.-N. Tran, M.~F. Hanif, A.~Tolli, and M.~Juntti, ``Fast converging algorithm
  for weighted sum rate maximization in multicell {MISO} downlink,''
  \emph{{IEEE} Signal Process. Lett.}, vol.~19, no.~12, pp. 872--875, 2012.

\bibitem{zhang2018performance}
C.~Zhang, Y.~Jing, Y.~Huang, and L.~Yang, ``Performance analysis for massive
  {MIMO} downlink with low complexity approximate zero-forcing precoding,''
  \emph{{IEEE} Trans. Commun.}, vol.~66, no.~9, pp. 3848--3864, 2018.

\bibitem{petersen2008matrix}
K.~B. Petersen, M.~S. Pedersen \emph{et~al.}, ``The matrix cookbook,''
  \emph{Technical University of Denmark}, vol.~7, no.~15, p. 510, 2008.

\bibitem{hjorungnes2007complex}
A.~Hjorungnes and D.~Gesbert, ``Complex-valued matrix differentiation:
  {T}echniques and key results,'' \emph{{IEEE} Trans. Signal Process.},
  vol.~55, no.~6, pp. 2740--2746, 2007.

\end{thebibliography}
\end{document}